\documentclass[11pt,DIV=12,a4paper,abstract=true,footinclude=false,headinclude=true]{scrartcl} 

\usepackage[utf8]{inputenc} 
\usepackage[T1]{fontenc} 
\usepackage{lmodern} 
\usepackage[protrusion=true,expansion=true]{microtype} 

\usepackage{mathtools} 
\usepackage{amsmath,amssymb,amsthm,amsfonts} 
\usepackage{bbm} 

\usepackage{array} 
\usepackage{booktabs} 
\usepackage{threeparttable} 
\usepackage{multirow} 
\usepackage{caption} 

\usepackage[english]{babel} 

\usepackage{graphics,graphicx} 

\usepackage{bm} 

\usepackage{enumerate} 
\usepackage[shortlabels]{enumitem} 

\usepackage{lipsum} 
\usepackage{comment} 

\usepackage[dvipsnames]{xcolor} 
\usepackage{dsfont} 
\usepackage[mathscr]{euscript} 
\usepackage{authblk}
\usepackage[]{listofsymbols}
\usepackage[round]{natbib}
\usepackage{xargs} 

\DeclareUnicodeCharacter{00A0}{ }

\newcommandx{\revision}[2][1=]{\todo[linecolor=ForestGreen,backgroundcolor=ForestGreen!25,bordercolor=ForestGreen,#1]{#2}} 

\newcommand{\from}{{:}\penalty\binoppenalty\mskip\thickmuskip}
\newcommand{\indicatorset}[1]{\mathds{1}_{#1}}

\newcommand{\toinucp}{\xrightarrow{{\lower2pt\hbox{$\!\mathrm{\scriptstyle ucp}\!$}}}}

\opensymdef
\newsym[Set of all natural numbers excluding zero]{na}{\mathbb{N}}
\newsym[Set of all natural numbers including zero]{nawz}{\na_{0}}
\newsym[Set of all real numbers]{re}{\mathbb{R}}
\newsym[Set of all nonnegative real numbers]{replus}{\re_{+}}
\newsym[Base $\sigma$-field]{Fcal}{\mathcal{F}}
\newsym[Base filtration]{FF}{\mathbb{F}}
\newsym[Base probability measure]{PP}{\mathbb{P}}
\newsym[Alternative probability measure]{QQ}{\mathbb{Q}}
\newsym[Borel $\sigma$-field]{B}{\mathcal{B}}
\newsym[Topology]{T}{\mathcal{T}}
\newsym[Space of all real valued semimartingales]{SM}{\mathscr{S}}
\newsym[Space of all real valued caglad processes]{LD}{\mathscr{L}}
\newsym[Space of all real valued cadlag processes]{DD}{\mathscr{D}}
\newsym[Set of controls]{MM}{\mathbb{M}}
\newsym[Set of controlled processes]{CP}{\mathcal{NN}(\Theta, \psi)}
\newsym[Process that generates the (augmented) base filtration]{Y}{Y}
\newsym[Integrands, trading strategies]{X}{X}
\newsym[Semimartingale integrators]{SI}{S}
\closesymdef

\theoremstyle{plain}
\newtheorem{theorem}{Theorem}[section]
\newtheorem{proposition}[theorem]{Proposition}

\theoremstyle{definition}
\newtheorem{definition}[theorem]{Definition}

\theoremstyle{remark}
\newtheorem{remark}[theorem]{Remark}

\DeclareMathOperator*{\argmin}{arg\,min}

\makeatletter

\DeclareOldFontCommand{\bf}{\normalfont\bfseries}{\mathbf}

\newcommand{\makeoverbar}[7]{%
\setbox0=\hbox{$\m@th#2\mkern#5mu{{}#3{}}\mkern#6mu$}%
\setbox1=\null \dimen@=#4\fontdimen8#13 \dimen@=3.5\dimen@
\advance\dimen@ by \ht0 \dimen@=-#7\dimen@ \advance\dimen@ by \wd0
\ht1=\ht0 \dp1=\dp0 \wd1=\dimen@
\dimen@=\fontdimen8#13 \fontdimen8#13=#4\fontdimen8#13
\rlap{\hbox to \wd0{$\m@th\hss#2{\overline{\box1}}\mkern#5mu$}}
\fontdimen8#13=\dimen@}

\def\makeunderbar#1#2#3#4#5#6#7{%
\setbox0=\hbox{$\m@th#2\mkern#5mu{{}#3{}}\mkern#6mu$}%
\setbox1=\null \dimen@=#4\fontdimen8#13 \dimen@=3.5\dimen@
\advance\dimen@ by \dp0 \dimen@=-#7\dimen@ \advance\dimen@ by \wd0
\ht1=\ht0 \dp1=\dp0 \wd1=\dimen@
\dimen@=\fontdimen8#13 \fontdimen8#13=#4\fontdimen8#13
\rlap{\hbox to \wd0{$\m@th\hss#2{\underline{\box1}}\mkern#5mu$}}
\fontdimen8#13=\dimen@}

\patchcmd{\@maketitle}% <cmd>
  {{\usekomafont{date}{\@date \par}}%
    \vskip \z@ \@plus 1em}% <search>
  {}% <replace>
  {}{}% <success><failure>

\def\thanks#1{\protected@xdef\@thanks{\@thanks
        \protect\footnotetext{#1}}}

\makeatother

\allowdisplaybreaks
\usepackage[pagebackref=false]{hyperref} 
\hypersetup{
	final=true, 
	pdfstartview = {FitH}, 
	pdftitle={Optimal Risk Mitigation by Deep Reinsurance}, 
	pdfauthor={A. Arandjelovic, J. Eisenberg}, 
	pdfcreator={A. Arandjelovic, J. Eisenberg}, 
	pdfsubject={}, 
	pdfkeywords = {}, 
	pdfdisplaydoctitle=true, 
	colorlinks=true, 
	linkcolor=blue, 
	citecolor=blue, 
	filecolor=blue, 
	urlcolor=blue
} 

\setkomafont{pageheadfoot}{\normalfont}

\usepackage[%
	automark,
	autooneside=false,
	headsepline=0.4pt
]{scrlayer-scrpage} 
\clearpairofpagestyles 
\setkomafont{disposition}{\bfseries} 
\usepackage[singlespacing]{setspace} 
\KOMAoptions{DIV=last} 

\let\oldbibliography\thebibliography 
\renewcommand{\thebibliography}[1]{%
  \oldbibliography{#1}%
  \setlength{\itemsep}{\smallskipamount}
  \setlength{\parskip}{0pt plus 1pt}} 

\title{Optimal Risk Mitigation by \\ Deep Reinsurance} 

\author[1]{Aleksandar~Arandjelovi\'{c}} 
\author[2]{Julia~Eisenberg%
\thanks{Address correspondence to Aleksandar~Arandjelovi\'{c}. Current address: Department of Mathematics, ETH Zurich, R\"amistrasse 101, 8092 Zurich, Switzerland; e-mail:\ \href{mailto:aleksandar.arandjelovic@math.ethz.ch}{aleksandar.arandjelovic@math.ethz.ch}.}} 

\affil[1]{Institute for Statistics and Mathematics, Vienna University of Economics and Business \newline Welthandelsplatz 1, 1020 Vienna, Austria} 
\affil[2]{Research Unit of Financial and Actuarial Mathematics, TU Wien \newline Wiedner Hauptstra{\ss}e 8--10, 1040 Vienna, Austria} 

\begin{document}

\maketitle 

\vspace{-36.5pt} 

\begin{abstract} 
We consider an insurance company which faces financial risk in the form of insurance claims and market-dependent surplus fluctuations. 
The company aims to simultaneously control its terminal wealth (e.g., at the end of an accounting period) and the ruin probability in a finite time interval by purchasing reinsurance. 
The target functional is given by the expected utility of terminal wealth perturbed by a modified Gerber--Shiu penalty function. 
We solve the problem of finding the optimal reinsurance strategy and the corresponding maximal target functional via neural networks. 
The procedure is illustrated by a numerical example, where the surplus process is given by a Cramér--Lundberg model perturbed by a mean-reverting Ornstein--Uhlenbeck process. 

\par\smallskip\noindent\textbf{Keywords:} Reinsurance; deep learning; perturbed risk process; multi-objective optimization. 

\par\smallskip\noindent\textbf{JEL codes:} C61; C63; G22. 
\end{abstract} 

%%%%%%%%%%%%%%%%%%%%%%%%%%%%%%----------------------------------------%%%%%%%%%%%%%%%%%%%%%%%%%%%%%% 
\section{Introduction}\label{sec:introduction} 
In 1903 Filip Lundberg suggested modeling the surplus of an insurance company by a constant drift minus a compound Poisson process with independent, identically distributed nonnegative jumps \citep{lundberg03kollektivrisker}. 
The drift is interpreted as the premium rate, the jumps represent the claim sizes, whereas the Poisson process counts the claims. 
This setting, widely known as the classical risk model or Cram\'{e}r--Lundberg model, gives a clear but simplified picture of the insurer's balance. 

For an insurance company, claims are not the only source of uncertainty. 
For example, the increase or decrease in the number of customers (see \citet{braunsteins23fluctuating}) or a random interest rate (see \citet{eisenberg15optimal}) will impact the risk process. 
Also, reputational considerations, investments in positively correlated financial markets, and correlations between different business branches and collectives of insured will play a considerable role. 
Recently, models involving a dependence between the actuarial business and financial markets offering investment possibilities have been considered; see, for instance, \citet{ceci22optimal}, \citet{leimcke20bayesian}, and references therein. 

An important modification of the classical risk model is due to \citet{gerber70renewal}, who suggested including an additional source of uncertainty -- a Brownian motion. 
This new, perturbed classical risk model can better account for reality while still being a one-dimensional Markov process. 
The latter property makes the perturbed process a popular model for the surplus of an insurance company, see, for instance, \citet{dufresne91risk}, \citet{tsai01diffusion}, \citet{cheung23joint} and references therein. 

However, by adding an additional source of uncertainty, in many optimization settings, the calculations and proofs become much more complicated, resulting in the use of the viscosity approach, see for instance \citet{eisenberg15optimal}. 
The viscosity approach allows to find the optimal strategy numerically, since the corresponding Hamilton--Jacobi--Bellman equation can be tackled using the finite difference method. 
To avoid this approach, one may discretize the surplus process, thereby allowing controls only at discrete time points. 
In particular, for a finite time horizon, this method has the advantage that all control strategies can be written in feedback form, i.e.\ depending on the finite number of the observed state values. 

In actuarial control theory, after the model for the surplus has been chosen, the question arises which risk measure will be considered as a target to optimize. 
The most famous and extensively studied risk measure, suggested by \citet{lundberg03kollektivrisker}, is the ruin probability, i.e.\ the probability that the surplus becomes negative in finite time. 
A vast number of results have appeared over the last century concerning the minimization of the ruin probability in different settings. 
We refer the interested reader to \citet{schmidli08control}, \citet{asmussen10ruin} and references therein. 
Alternative risk measures that have been considered over the last decades include in particular expected dividends, expected utility of terminal wealth, and expected capital injections, see, for instance, \citet{albrecher17reinsurance}, \citet{avanzi09strategies}, \citet{albrecher09optimality} and references therein. 

Measuring the utility of the terminal wealth was first suggested by \citet{borch61utility} and had since become an important risk indicator in insurance mathematics. 
The wealth at some finite time $T$ -- for instance, the time of a regulatory check, can provide useful clues about the company's wellbeing. 
Including the ruin probability in the target functional has been considered, for instance, in \citet{hipp18value}, where the main target is to maximize dividends, in the more recent paper \cite{albrecher25dividend}, and in \citet{thonhauser07dividend}, where the time value of ruin is taken into account. 
However, the multi-objective goal of simultaneously optimizing a risk measure (such as, in our case, the expected utility of terminal wealth) along with the ruin probability remains largely unexplored. 

In the present manuscript, we consider a general extension of the perturbed risk model that, to the best of our knowledge, has not been studied in the existing actuarial literature. 
The surplus of an insurance company is modelled by a jump process perturbed by a general diffusion (not necessarily a Brownian motion) on an interval $[0,T]$ with a deterministic time horizon $T$. 
This implies that the problem we consider is 3-dimensional and depends on the time to maturity, on the state of the jump process and on the state of the diffusion. 
The functional to maximize is given by the expected utility of terminal wealth perturbed by a modified Gerber--Shiu \citep{gerber98ruin} penalty function. 
It is optimized over the class of reinsurance policies, which are arguably the most popular type of controls in the literature (with investments, dividends and capital injections being common alternatives). 
A substantial body of literature exists on utility maximization and ruin minimization with reinsurance strategies, with notable contributions including \citet{schmidli02ruin}, \citet{promislow05ruin}, \citet{bai08optimal}, \citet{schmidli01proportional} and \citet{taksar03optimal}. 

The role of the penalty function in this paper is twofold. 
It rewards the insurer if the surplus remains positive at all times (i.e.\ in the case of no ruin), while a negative surplus is penalized. 
In addition, one can opt for different weights for the expected utility and for the penalty function, depending on the individual preferences of the insurer. 
This means that we allow the insurer to prioritize their immediate needs: higher utility of the terminal surplus with higher risk or a safer play. 

As it seems unlikely that this problem can be solved explicitly, we seek for the optimal strategy in the class of feedback controls using machine learning techniques. 
More specifically, we use neural networks. 
Neural networks have become a popular tool in actuarial risk management, having been applied to mortality modelling, claims reserving, non-life insurance pricing and telematics. 
For a survey on recent advances of artificial intelligence in actuarial science, see \citet{richman21review} and references therein. 

The task of finding optimal reinsurance (and dividend) strategies with neural networks has been studied in \citet{cheng20optimal} and \citet{jin21hybrid}. 
There, the authors develop a hybrid Markov chain approximation-based iterative deep learning algorithm to maximize expected dividends under consideration of the time value of ruin. 
In contrast to \citet{cheng20optimal} and \citet{jin21hybrid}, we include the ruin probability explicitly as a risk measure (besides the expected utility of terminal wealth) that is optimized. 
Moreover, we formulate our optimization problem as an empirical risk minimization problem, which can be solved efficiently by stochastic gradient descent methods, even in highly complex model settings. 

The primary contributions of this paper are as follows: 
\begin{enumerate} 
\item We introduce a novel framework for optimizing reinsurance strategies using deep learning techniques in order to maximize a target functional comprising a utility function penalized by an extended Gerber--Shiu function. 
The proposed method allows the insurer to balance between maximizing the expected utility of terminal wealth, and minimizing the probability of ruin. 
\item By drawing connections to binary classification problems and surrogate loss functions, we demonstrate how the optimization problem can be solved by empirical risk minimization, a method that, when combined with stochastic gradient descent, is particularly useful for optimizing neural networks. 
\item We illustrate our proposed methodology by a numerical example, where the surplus process is given by a Cramér--Lundberg model perturbed by a mean-reverting Ornstein--Uhlenbeck process. 
Our findings demonstrate the effectiveness of our method in finding optimal reinsurance strategies, and highlight the large scope of the approach. 
\end{enumerate} 

The paper is organized as follows. 
Section \ref{sec:model} gives a mathematical description of the considered model. 
Section \ref{sec:algopolicies} introduces algorithmic reinsurance strategies. 
Section \ref{sec:numerics} exemplifies our approach with numerical experiments. 
Section \ref{sec:conclusion} concludes. 

%%%%%%%%%%%%%%%%%%%%%%%%%%%%%%----------------------------------------%%%%%%%%%%%%%%%%%%%%%%%%%%%%%% 
\section{Model description}\label{sec:model} 
We consider an insurer with a deterministic finite planning horizon $T>0$. 
The insurer manages a portfolio of risks that generate premium payments. 
To mitigate potential large financial losses from unexpectedly high claim frequencies or sizes, the insurer can enter into reinsurance agreements. 
For a given number of time steps $n \in \na$, these agreements are re-negotiated at time points $(t_{i})_{i=0}^{n}$, $0=t_{0}<t_{1}<\hdots<t_{n-1}<T$ for the coverage period $(t_{i}, t_{i+1}]$, where we set $t_{n}=T$. 

Our study is based on a probability space $(\Omega, \Fcal, \PP)$. 
The flow of information is modeled by an $\mathbb{R}^{r}$-valued stochastic process $Y$, where $r \in \na$ is a fixed dimension. 
The process $Y$ induces a filtration $\FF = (\Fcal_{i})_{i=0}^{n}$ of $\Fcal$, allowing for the possibility that $\Fcal_{n}\neq\Fcal$. 
That is, there might be some information that is not revealed even at maturity $T$. 
We assume that $\mathbb{R}^{r}$ is endowed with the Borel-$\sigma$-algebra $\mathcal{B}_{\re^r}$. 
All stochastic processes are indexed via the discrete time points $(t_{i})_{i=0}^{n}$. 

Let $p = (p_{i})_{i=0}^{n}$ denote the $\FF$-adapted premium process. 
The payment $p_{i}$ ensures insurance coverage over the period $(t_{i}, t_{i+1}]$. 
As in our numerical illustrations in Section~\ref{sec:numerics}, $p_{i}$ might be computed according to the expected value principle with a positive safety loading. 
Alternatively, $p_{i}$ can be calculated using any other known premium calculation principle such as, for instance, the standard deviation or zero-utility principle. 
The premium may also depend on the number and size of the previously occurred claims. 
However, for the sake of clarity, we concentrate on the expected value principle in order to better explain the features of our model and leave further extensions to future research. 
Since the time horizon $T$ is assumed to be finite, we may assume, without loss of generality, that $p_{n}=0$. 

Inspired by the Cram\'{e}r--Lundberg model, let $N = (N_{i})_{i=0}^{n}$ be an $\FF$-adapted, $\na_{0}$-valued and increasing process with $N_{0} = 0$, and $N_{i}$ represents the number of claims up to time $t_{i}$. 
The $\replus$-valued insurance claims are denoted by $(Z_{i})_{i \ge 1}$. 
We also consider a real-valued, $\FF$-adapted process $L = (L_{i})_{i=0}^{n}$ which represents random fluctuations, such as small claims and variations in premium income. 

\begin{remark} 
A distinguishing feature is that $p$, $N$, $(Z_{i})_{i \ge 1}$ and $L$ are not assumed to be independent. 
\end{remark} 

The reinsurance agreement is characterized by a reinsurance strategy $b=(b_{i})_{i=0}^{n}$. 
For illustrative purposes, we assume the agreement to be proportional, that is, $b$ is an $\FF$-adapted process with values in $[0, 1]$. 
Here, $b_{i}$ represents the retention level, and $(1-b_{i})$ is the proportion of claims covered by the reinsurer during the period $(t_{i}, t_{i+1}]$. 
For notational convenience, we set $b_{n} = 1$. 
The reinsurance premium is given by the process $c = (c_{i})_{i=0}^{n}$. 
We assume that $c_{i} = c(b_{i})$ for some continuous cost function $c \from [0, 1] \to \re$. 

For each reinsurance strategy $b$, the surplus process $X^{b} = (X_{i}^{b})_{i=0}^{n}$ is defined by $X_{0}^{b}=x$, and 
\begin{equation}\label{eq:surplus} 
    X_{i+1}^{b} = X_{i}^{b} + p_{i} - c(b_{i}) + L_{i} - b_{i} \sum_{j=N_{i}+1}^{N_{i+1}}Z_{j}, \quad i=0,1,\hdots,n-1, 
\end{equation} 
where $x\in\re$ is the initial capital. 
The surplus process without reinsurance is denoted $X = (X_{i})_{i=0}^{n}$. 
Note that the sum in Equation~\eqref{eq:surplus} could be empty, namely when $N_i = N_{i+1}$ (i.e.\ no new claims during the time interval $(t_i, t_{i+1}]$). 
In this scenario we convene that the sum is zero. 

Since the seminal paper by \citet{gerber70renewal}, introducing fluctuations into the classical Cramér--Lundberg model has become a significant area of research within the actuarial community. 
To the best of our knowledge, no prior work has considered perturbing the classical risk process using a general diffusion process, other than the standard Brownian motion. 
In our approach, we opted to increase the model's complexity by incorporating a disturbance term given by a general diffusion instead of, for instance, expanding the optimization problem through a multidimensional reinsurance deductible. 

Our problem formulation is kept very general, so that various notable models can be considered as special cases. 
For example, $N$ might be a self-exciting process observed at discrete time points. 
As in Section~\ref{sec:numerics}, $L$ could be (the discretization of) an Ornstein--Uhlenbeck process. 
Additionally, let us note that the formulation of this section extends naturally to the multi-dimensional case with, for instance, multiple correlated business lines. 

The insurer's preference is described by a continuous utility function $u \from \re \to \re$. 
Here, we assume that the utility function is chosen such that $\mathbb{E}[u(X_{n})]$ is finite. 
For example, in Section \ref{sec:numerics} we will choose an exponential utility function in conjunction with exponentially distributed claims. 
If one were to choose Pareto-distributed claims -- a popular choice in the literature -- then another utility function is required due to the heavy tails of the Pareto distribution. 

\begin{definition} 
A reinsurance strategy $b$ is called \textit{admissible} if $\mathbb{E}[u(X_{n}^{b})]$ is finite. 
The set of all admissible reinsurance strategies is denoted by $\mathcal{A}$. 
\end{definition} 

In our model, the insurer aims to optimize the expected utility of terminal wealth while considering the probability of ruin across all admissible reinsurance strategies. 
This is done by incorporating a penalty term whose strength is expressed by a parameter $\beta \in [0,1]$. 
That way, the target is to solve the following optimization problem 
\begin{equation}\label{goal} 
\sup_{b \in \mathcal{A}} \bigg\{ \beta \cdot \mathbb{E}[u(X_{n}^{b})] - (1-\beta) \cdot \mathbb{P}(\min_{0\le i\le n} X_{i}^{b} < 0) \bigg\}. 
\end{equation} 

\begin{remark} 
The choice of our objective function in~\eqref{goal} that combines utility of terminal wealth and the probability of ruin is motivated as follows: First of all, by using convex combinations, we include the two boundary cases $\beta=1$ and $\beta=0$ which correspond to the pure expected utility maximization and ruin probability minimization problems, respectively. 
More generally, the convex combination in~\eqref{goal}, is a linear scalarization of a multi-objective optimization problem, i.e.\ the original problem with two objectives -- the expected utility of terminal wealth and the ruin probability -- is transformed into a single-objective optimization problem. 
Our approach can also accommodate other objectives alongside those discussed in this paper, we leave the investigation of other objectives to future research. 
It is important to note that Equation~\eqref{goal} does not account for the different scales of the individual goals. This means that a weight of $\beta = 50\%$ will not necessarily yield an optimal reinsurance strategy that appropriately balances both goals by $50\%$ each. 
This is a common issue that arises when solving multi-objective optimization problems through scalarization. 
It is a subject of future research to address and resolve this challenge. 
We also refer to \citet{feinstein24frontier}, where the efficient frontier, i.e.\ the set of all optimal solutions for a multi-objective optimization problem, is approximated via neural networks. 
\end{remark} 

Note that the desired optimal strategy is not intended to be implemented in a literal manner. 
Instead, it functions as a risk measure, providing insight into the financial soundness or vulnerability of a firm. 
Should the prospect of ruin be considered acceptable, this would suggest that achieving the optimal value requires taking on significant risk that may align with the strategic preferences of the firm's management. 
However, such a risk-loving attitude would not be tolerable from a regulatory standpoint. 
Introducing a lower bound on the parameter $\beta$ could help to align the model with regulatory frameworks, which can be compared to the minimum capital requirement\footnote{For more information on the minimum capital requirement see\\ \url{https://www.bmf.gv.at/en/topics/financial-sector/finance-captital-markets-eu/solvency-ii.html}}. 

From a technical standpoint, it may be argued that the exponential utility function remains well-defined even in the presence of negative values, although this entails a corresponding reduction in utility. 
As a result, both the probability of ruin and the risk of negative surplus exert a direct influence on the strategy, irrespective of the utility function's central role. 
In this context, the emergence of negative capital does not represent a simple trade-off; rather, it is inherently penalized by both the utility formulation and the associated likelihood of ruin. 

We aim to solve the optimization problem \eqref{goal} using empirical risk minimization and stochastic approximation, a classical concept in statistical learning theory which is explained below. 

Note that our goal is not to compare the neural network approach with other methods, nor to evaluate their algorithmic efficiency. 
Instead, we seek to highlight the broad applicability of the proposed methodology across a wide range of problems. 

Given a set of training data points $\{(Y^{j}, Z^{j})\, |\, j=1,2,\hdots,m\}$ that are independent and identically distributed, and a parametrized family of predictor functions $\{f_{\theta} \colon \theta \in \Theta\}$, the goal is to minimize the empirical loss, 
\begin{equation*} 
\theta^{*} = \argmin_{\theta \in \Theta} \frac{1}{m} \sum_{j=1}^{m} \ell(f_{\theta}(Y^{j}), Z^{j}), 
\end{equation*} 
where $\ell$ denotes a loss function. 
If $\ell(f_{\theta}(Y^{j}), Z^{j})$ is differentiable with respect to $\theta$, the gradient descent method can be used to minimize the empirical loss starting from an initial guess $\theta_{0}$ and iteratively updating
\begin{equation*} 
\theta_{k+1} = \theta_{k}- \frac{\eta}{m} \sum_{j=1}^{m} \nabla_{\theta_{k}} \ell(f_{\theta_{k}}(Y^{j}), Z^{j}), \quad k=1, 2, \hdots, 
\end{equation*} 
until a pre-determined termination criterion is reached. 
Here, $\eta > 0$ denotes a learning rate, and $\nabla_{\theta} \ell(f_{\theta}(Y^{j}), Z^{j})$ denotes the gradient of $\ell(f_{\theta}(Y^{j}), Z^{j})$ with respect to $\theta$. 

For very large sample sizes, where computing the average gradient is numerically expensive, \textit{stochastic gradient descent} (SGD) is a more efficient alternative, updating parameters based on individual gradients $\nabla_{\theta} \ell(f_{\theta}(Y^{j}), Z^{j})$. 
However, due to the noisy nature of single-gradient updates, \textit{mini-batch} stochastic gradient descent is often preferred. 
This method updates parameters based on the average gradient over a subset of the training points, balancing computational efficiency and stability. 

Note that the loss $\ell(f_{\theta}(Y^{j}), Z^{j})$ is assumed to be differentiable with respect to $\theta$. 
However, in our optimization problem, the ruin probability imposes numerical challenges due to the non-smooth nature of the indicator function. 
To address this issue, Subsection~\ref{binary-classification} explores connections to binary classification by replacing the indicator function with a surrogate loss function. 
This substitution makes the problem tractable via empirical risk minimization. 
Our surrogate loss model can be interpreted as a generalized version of a Gerber--Shiu penalty function, tailored to the multi-objective setting considered here and, to the best of our knowledge, not previously used in this form. 

%%%%%%%%%%%%%%%%%%%%%%%%%%%%%%----------------------------------------%%%%%%%%%%%%%%%%%%%%%%%%%%%%%% 
\subsection{Ruin probability and binary classification problems}\label{binary-classification} 
Ruin probabilities can be expressed in terms of expectations. 
To this end, given $b \in \mathcal{A}$, consider the map $F_{b}$ defined by $F_{b}(Y) = \min_{0\le i\le n}X_{i}^{b}$. 
Then one can write 
\begin{equation}\label{binary-loss} 
\mathbb{P}(\min_{0\le i\le n} X_{i}^{b} < 0) = \mathbb{P}(F_{b}(Y) < 0) = \mathbb{E}[\indicatorset{(-\infty,0)}(F_{b}(Y))], 
\end{equation} 
where $\indicatorset{A}(x)$ denotes the indicator function over the set $A$, i.e.\ $\indicatorset{A}(x) = 1$ for $x \in A$ and $\indicatorset{A}(x) = 0$ otherwise. 

Minimizing the ruin probability in \eqref{binary-loss} then amounts to finding the optimal $b\in\mathcal{A}$, such that the mapping $F_{b}$ classifies as many data points $Y(\omega)$ as possible not as ruin. 
In other words, finding the optimal reinsurance strategy $b\in\mathcal{A}$ can be identified with the equivalent task of finding the optimal classifier $F_{b^{*}} \in \{ F_{b} \colon b \in \mathcal{A} \}$ for the binary classification problem where one seeks to map all data points $\{ Y(\omega) \colon \omega \in \Omega \}$ to a non-ruin event. 
Using $Y$ as an argument to $F_{b}$ is justified by the Doob--Dynkin representation theorem, which states that all $\FF$-adapted processes can be written as functions of $Y$. 

The main issue with the indicator function is twofold: (1) It has a jump in zero, which is problematic if the distribution of the minimal wealth $F_b(Y)$ has positive probability mass in zero, and (2) it is not differentiable in zero, and the derivatives away from zero are also equal to zero. 
Therefore, optimizing \eqref{binary-loss} by using deep learning tools and empirical risk minimization, particularly minibatch stochastic gradient descent, becomes quite challenging. 
A remedy to this issue is to employ a \textit{surrogate loss function $g$} instead of the indicator function: 
\begin{equation}
\mathbb{E}[\indicatorset{(-\infty,0)}(F_{b}(Y))] \approx \mathbb{E}[g(F_{b}(Y))]. 
\end{equation} 
The surrogate loss function $g_{\gamma}(x) = 0.5 + 0.5 \tanh(-\gamma x)$ for $\gamma\in\{1,10,100\}$ is presented in Figure~\ref{penalty-functions}. 
This function will also be used (with $\gamma=10$) in the numerical study in Section~\ref{sec:numerics}. 
For theory on surrogate loss functions for binary classification problems, see \citet{bartlett06convexity,nguyen09surrogate,reid10composite}. 

\begin{figure}[htbp] 
\centering 
\captionsetup{width=0.75\textwidth,format=plain} 
\includegraphics[width=0.75\linewidth]{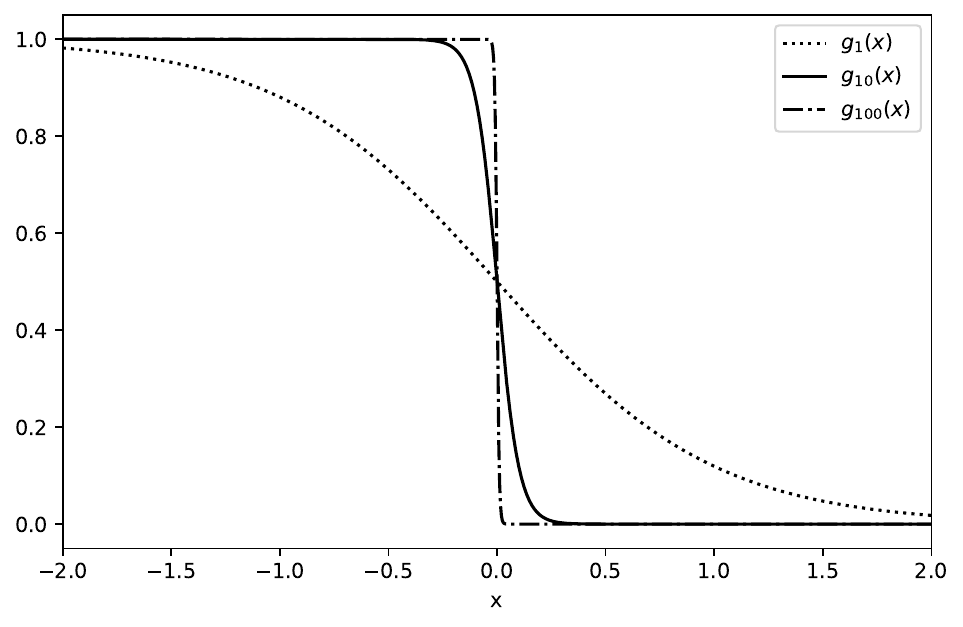}
\caption{Surrogate loss functions $g_{\gamma}$ for various choices of $\gamma$.} 
\label{penalty-functions} 
\end{figure} 

Replacing the ruin probability in \eqref{goal} by the expected surrogate loss of $F_{b}(Y)$ with respect to $g_{\gamma}$ imposes a penalty on the expected utility, which can be seen as a generalized Gerber--Shiu function \citep{gerber98ruin}. 
Indeed, our surrogate loss function allows for an approximation of the ruin probability under mild assumptions. 

\begin{proposition}\label{surrogate-limit} 
Given $b \in \mathcal{A}$, let $(g_{n})_{n \in \na}$ be a uniformly bounded sequence of functions such that, $\PP$-almost surely, $g_{n}(F_{b}(Y)) \to \indicatorset{(-\infty,0)}(F_{b}(Y))$. 
Then, 
\begin{equation*}
\lim_{n\to\infty} \mathbb{E}[g_{n}(F_{b}(Y))] = \mathbb{P}(\min_{0 \le i\le n} X_{i}^{b} < 0). 
\end{equation*} 
\end{proposition} 

\begin{proof} 
This is a direct consequence of Lebesgue's dominated convergence theorem. 
\end{proof} 

\begin{remark} 
For our choice $g_{\gamma}$ of surrogate loss function as presented in Figure~\ref{penalty-functions}, Proposition~\ref{surrogate-limit} is applicable as $\gamma \to \infty$ if we assume that $F_{b}(Y)$ does not have a point mass in zero, i.e.\ we require $\PP(F_{b}(Y) = 0) = 0$. 
The reason is that $g_{\gamma}(0) = 1/2$, which is the only point where $g_{\gamma}$ does not converge to the indicator function $\indicatorset{(-\infty,0)}$. 
\end{remark} 

%%%%%%%%%%%%%%%%%%%%%%%%%%%%%%----------------------------------------%%%%%%%%%%%%%%%%%%%%%%%%%%%%%% 
\section{Algorithmic reinsurance policies}\label{sec:algopolicies} 
We propose to solve \eqref{goal} numerically via \textit{algorithmic reinsurance policies}. 
These policies determine retention levels using neural networks that observe information to generate decisions. 
This approach is inspired by recent successful applications in quantitative finance and actuarial science. 
Popular use-cases are hedging, optimal stopping, model calibration, and scenario generation \citep{buehler19deep,becker19deep,horvath21volatility,wiese20gan}. 

Theorems that establish approximations in function space via neural networks are usually referred to as \textit{universal approximation theorems} (UAT); notable contributions include \citet{cybenko89approximation} and \citet{hornik91approximation}. 
These theorems establish the topological density of sets of neural networks in various topological spaces. 
One speaks of the universal approximation property \citep{kratsios21uap} of a class of neural networks. 
Unfortunately, these theorems are usually non-constructive. 
To numerically find optimal neural networks, one typically combines backpropagation (see, for example, \citet{rumelhart86representations}) with ideas from stochastic approximation \citep{robbins51approximation,kiefer52estimation,dvoretzky56approximation}. 

\begin{definition}[Deep feedforward neural network]\label{neural-network} 
Let $\psi \from \re \to \re$ be a bounded, measurable and non-constant map. 
Given $k, l, n \in \na$, we denote by $\mathcal{NN}_{k,l,n}(\psi)$ the set of neural networks with an $n$-dimensional input layer, one neuron with identity activation function in the output layer, $k$ hidden layers, and at most $l$ nodes with $\psi$ as activation function in each hidden layer (cf. \citet[Definition~3.1]{kidger20narrow}). 
We call elements from $\mathcal{NN}_{k,l,n}(\psi)$ \textit{deep feedforward neural networks}, or simply \textit{deep neural networks}. 
\end{definition} 

\begin{definition}[Algorithmic reinsurance policy]\label{algo-strat} 
We denote by $\mathcal{A}^{\mathrm{nn}}$ the set of all proportional reinsurance strategies $b$ that satisfy $b_{i} = \sigma \circ f(Y_{0}, Y_{1}, \hdots, Y_{i})$ for some $f \in \mathcal{NN}_{k_{i}, l_{i}, (i+1)r}(\psi)$ with $k_{i}, l_{i} \in \na$, for each $i=0,1,\hdots,n$. 
Here, $\sigma \from \re \to (0,1)$ denotes the logistic function given by $\sigma(x) = \exp(x) / (1+\exp(x))$. 
\end{definition} 

Definition \ref{neural-network} provides one example of a class of neural networks we can use, but other choices are possible. 
In light of Definition \ref{algo-strat}, one could consider for example recurrent neural networks or long short-term memory (LSTM) networks (see for example \citet{hochreiter97lstm}). 
We also refer to UATs for deep, narrow networks \citep{kidger20narrow} and for randomized neural networks \citep{huang06incremental}. 
Note that feedforward neural networks are usually defined with a linear readout map (as in Definition~\ref{neural-network}). 
In order to ensure that algorithmic reinsurance policies are valid proportional reinsurance policies assuming values in $[0,1]$, we apply the logistic function $\sigma$ to the network output in Definition~\ref{algo-strat}. 

In this paper, we restrict ourselves to proportional reinsurance strategies. 
However, the same techniques can be applied to other types of reinsurance strategies. 
For example, excess-of-loss (XL) policies could be written as $b_{i}(Z) = \min\{Z, f(Y_{0}, Y_{1}, \hdots, Y_{i})\}$ for some deep neural network $f$. 
We leave the investigation and comparison of results for various algorithmic reinsurance treaties to future research. 

For the sake of notational simplicity, given $b \in \mathcal{A}$, $\beta \in [0,1]$ and a smooth surrogate loss function $g$, let 
\begin{equation}
u_{\beta}(Y, b) = \beta u(X_{n}^{b}) - (1-\beta)g(F_{b}(Y)). 
\end{equation} 
From a numerical perspective, solving the optimization problem~\eqref{goal} via algorithmic reinsurance policies motivates the use of Monte Carlo methods. 
Here, we first generate a finite amount of data points $(Y^{j})_{j=1}^{m}$ and replace, for each $b \in \mathcal{A}$, the expectation and probability appearing in~\eqref{goal} by the empirical average 
\begin{equation}\label{MC-goal}  
\frac{1}{m} \sum_{j=1}^{m} u_{\beta}(Y^{j}, b). 
\end{equation} 
The summation over a finite number $m \in \na$ of data points allows us to re-interpret Equation~\eqref{MC-goal} as the expectation of $u_{\beta}(Y, b)$ over a measure which assigns equal probability $1/m$ to every outcome. 
This motivates the assumption that the underlying probability space $\Omega$ is finite, in which case all reinsurance strategies, including all algorithmic reinsurance policies, are admissible, i.e.\ in particular $\mathcal{A}^{\mathrm{nn}} \subset \mathcal{A}$. 
Moreover, since in this case every singleton set $\{\omega\}$ for each $\omega \in \Omega$ is measurable (for otherwise we could not assign the probability $1/m$ to it), it is also natural to assume that $\Omega$ is endowed with the power set $\mathcal{P}(\Omega)$ to form a measurable space. 

\begin{theorem}\label{algo-approx} 
Assume that $\Omega$ is a finite set, and that $\mathcal{F} = \mathcal{P}(\Omega)$. 
Then, for every $\beta \in [0,1]$ and $\varepsilon > 0$, there exists an algorithmic reinsurance policy $b^{\mathrm{nn}} \in \mathcal{A}^{\mathrm{nn}}$ such that 
\begin{equation} 
\mathbb{E}[u_{\beta}(Y,b^{\mathrm{nn}})] > \sup_{b \in \mathcal{A}} \mathbb{E}[u_{\beta}(Y,b)] - \varepsilon. 
\end{equation} 
\end{theorem} 

\begin{proof} 
The proof partially relies on some ideas from the proof of~\citet[Proposition~4.3]{buehler19deep}. 
Let $b^{*} \in \mathcal{A}$ be an $\varepsilon / 2$-optimal strategy, i.e.\ let $b^{*} \in \mathcal{A}$ be such that $\mathbb{E}[u_{\beta}(Y,b^{*})] > \sup_{b \in \mathcal{A}} \mathbb{E}[u_{\beta}(Y,b)] - \varepsilon/2$. 
Fix one time point $t_{i}$. 
Since $b^{*}$ is $\FF$-adapted, we have that $b_{i}^{*}$ is $\Fcal_{i}$-measurable.
Recall that $\Fcal_{i}$ is the smallest $\sigma$-algebra that makes $Y_{0}, Y_{1}, \hdots, Y_{i}$ measurable. 
An application of Doob--Dynkin's lemma implies the existence of a Borel-measurable map $f_{i} \from \re^{r \times (i+1)} \to [0,1]$ such that $b_{i}^{*} = f_{i}(Y_{0}, Y_{1}, \hdots, Y_{i})$. 

Let $\mu$ be the Borel probability measure that is induced by $Y_{0}, Y_{1}, \hdots, Y_{i}$ on $\re^{r \times (i+1)}$. 
Since $b_{i}^{*}$ is bounded (as it assumes values in $[0,1]$), we have $b_{i}^{*} \in L^{p}(\PP)$ for every $p > 0$ and thus, in particular, $f_{i} \in L^{2}(\mu)$. 

Consider the sequence $(f_{i}^{k})_{k \in \na}$ of functions given by $f_{i}^{k}(x) = 1-1/k$ if $f_{i}(x) =1$, $f_{i}^{k}(x) = 1/k$ if $f_{i}(x) = 0$, and $f_{i}^{k}(x) = f_{i}(x)$ otherwise. 
Clearly, $f_{i}^{k} \to f_{i}$ pointwise and thus, $f_{i}^{k}(Y_{0}, Y_{1}, \hdots, Y_{i}) \to f_{i}(Y_{0}, Y_{1}, \hdots, Y_{i}) = b_{i}^{*}$ pointwise. 
We repeat the same construction for all time points, and construct a sequence $(\hat{b}^{k})_{k \in \na}$ of reinsurance policies, where $\hat{b}_{i}^{k} = f_{i}^{k}(Y_{0}, Y_{1}, \hdots, Y_{i})$ for each $i = 0, 1, \hdots, n-1$. 
Since $\Omega$ is finite, 
\begin{equation*} 
\mathbb{E}[u_{\beta}(Y,\hat{b}^{k})] \to \mathbb{E}[u_{\beta}(Y,b^{*})], \quad k \to \infty. 
\end{equation*} 
We may therefore assume, without loss of generality, that $f_{i}$ takes values in a closed subset $[1/k,1-1/k]$ of $[0,1]$ that does neither contain $0$ nor $1$, for a sufficiently large $k \in \na$. 

The sigmoid function $\sigma$, being a continuous bijection with continuous inverse $\sigma^{-1} \from (0,1) \to \re$, establishes a homeomorphism between $\re$ and $(0,1)$. 
Since $f_{i} \in [1/k,1-1/k]$, it follows that $\sigma^{-1}\circ f_{i}$ is bounded, and thus $\sigma^{-1}\circ f_{i} \in L^{2}(\mu)$. 

The universal approximation theorem \citep{hornik91approximation} ensures the existence of a sequence $(\tilde{f}_{i}^{k})_{k \in \na}$ of shallow feedforward neural networks (i.e.\ feedforward neural networks with one hidden layer), such that $\tilde{f}_{i}^{k} \to \sigma^{-1}\circ f_{i}$ in $L^{2}(\mu)$, as $k \to \infty$.
But this implies, denoting $\tilde{b}_{i}^{k} = \tilde{f}_{i}^{k}(Y_{0}, Y_{1}, \hdots, Y_{i})$, that $\tilde{b}_{i}^{k} \to \sigma^{-1} \circ b_{i}^{*}$ in $L^{p}(\PP)$, as $k \to \infty$, and thus also in probability.
Upon passing to a subsequence, which we also denote $(\tilde{b}_{i}^{k})_{k \in \na}$, we may assume that convergence holds outside of a $\PP$-null set, in which case $\sigma \circ \tilde{b}_{i}^{k} \to b_{i}^{*}$, $\PP$-almost surely. 

If we repeat the above arguments for every $t_{i}$, we obtain a sequence of processes $(b^{k})_{k \in \na}$ given by $b_{i}^{k} = \sigma \circ \tilde{b}_{i}^{k}$, such that $b_{i}^{k} \to b_{i}^{*}$ for every $t_{i} \in T$, where convergence holds outside a $\PP$-null set. 

This implies that $X_{i}^{b^{k}} \to X_{i}^{b^{*}}$ for every $t_{i}$, where convergence holds outside a $\PP$-null set, and consequently $F_{b^{k}}(Y) \to F_{b^{*}}(Y)$, $\PP$-almost surely, which implies that $u(X_{n}^{b^{k}}) \to u(X_{n}^{b^{*}})$ and $g(F_{b^{k}}(Y)) \to g(F_{b^{*}}(Y))$, $\PP$-almost surely as $k \to \infty$. 
Since $\Omega$ is finite, 
\begin{equation*} 
\mathbb{E}[u_{\beta}(Y,b^{k})] \to \mathbb{E}[u_{\beta}(Y,b^{*})], \quad k \to \infty. 
\end{equation*} 
We can thus find $k^{*}$ large enough, such that $| \mathbb{E}[u_{\beta}(Y,b^{*})] - \mathbb{E}[u_{\beta}(Y,b^{k^{*}})] | < \varepsilon / 2$, and set $b^{\mathrm{nn}} = b^{k^{*}}$. 
\end{proof} 

\begin{remark} 
Theorem~\ref{algo-approx} provides theoretical justification for solving the optimization problem~\eqref{goal} via algorithmic reinsurance policies. 
However, it is non-constructive in the sense that it does not shed light on how the policy $b^{\mathrm{nn}}$ can be found in practice. 
In order to solve the problem via empirical risk minimization in the next section, we will proceed in two steps: 
\begin{enumerate} 
\item Approximate the indicator function used to form the ruin probability by a surrogate loss function. This step is justified by Proposition~\ref{surrogate-limit}. 
\item Solve the surrogate problem by stochastic approximation over the set $\mathcal{A}^{\mathrm{nn}}$ of algorithmic policies. 
\end{enumerate} 
The final result is subject to two sources of numerical error: (1) the approximation of the indicator function, and (2) the approximation of the solution to the surrogate problem with algorithmic policies. 
However, both of these two errors can be made arbitrarily small. 
\end{remark} 

\begin{remark} 
Classical universal approximation theorems which are formulated for feedforward neural networks usually assume a linear readout, which is not bounded. 
In order to obtain a valid proportional reinsurance policy, we need to guarantee that our algorithmic strategies assume values in $[0,1]$. 
In Theorem~\ref{algo-approx} we were able to do this with some simple tricks for our specific model setting. 
For a more general treatment of non-Euclidean universal approximation, see \citet{kratsios20euclidean}. 
\end{remark} 

%%%%%%%%%%%%%%%%%%%%%%%%%%%%%%----------------------------------------%%%%%%%%%%%%%%%%%%%%%%%%%%%%%% 
\section{Numerical experiments}\label{sec:numerics} 
In this section, we present a numerical example where the surplus is modeled by a Cram{\'e}r--Lundberg process perturbed by an Ornstein--Uhlenbeck (OU) process. 
Including the OU process adds complexity, making the problem more challenging and realistic. 
The OU process may, for instance, represent random fluctuations, such as small claims and variations in premium income. 

We assume that the claims $Z$ are independent, identically exponentially distributed with mean $\mu$, $N$ represents a discretized Poisson process with intensity $\lambda$, and $Z$ and $N$ are independent. 
The insurer charges premia according to the expected value principle with safety loading $\eta>0$, i.e.\ $p_{i}=(1+\eta)\lambda\mu(t_{i+1}-t_{i})$. 
Similarly, the reinsurer charges premia according to the expected value principle with safety loading $\theta>\eta$, i.e.\ $\tilde{p}_{i}=(1+\theta)\lambda\mu(1-b_{i})(t_{i+1}-t_{i})$. 
The process $L$ follows the dynamics 
\begin{equation}
L_{i+1} = L_{i} + \xi (\kappa - L_{i})(t_{i+1}-t_{i}) + \nu (t_{i+1}-t_{i}) \varepsilon_{i}, 
\end{equation} 
where $\varepsilon_{i}$ are independent, identically distributed shocks with $\mathcal{L}(\varepsilon_i) = \mathcal{N}(0,1)$. 
We assume that the initial capital is positive to avoid starting in ruin. For our simulations, we take an arbitrarily chosen initial value of 1. 

The utility function is of exponential type, $u(x) = -\exp(-\alpha x)$, with risk-aversion coefficient $\alpha$. 
Table \ref{base-parameters} summarizes the parameters for our base model. 
The values for $\lambda$, $\mu$, $\eta$ and $\theta$ are taken from \citet{schmidli01proportional}. 
The choice $\beta=0.4$ is illustrative and serves to generate a non-trivial trade-off between utility and ruin; other choices of $\beta$ lead to different points in Figure~\ref{trade-off-front}. 
For the neural network, we chose a two-hidden-layer topology with hyperbolic tangent (tanh) as activation function in the hidden layers and logistic activation $\sigma$ in the output layer. 
The hidden layers contain 32 nodes each. 
The neural network takes the surplus as input; additional inputs were tested but provided negligible improvements. 
All computations were performed using Python, using the Keras deep learning API for constructing and training the neural networks. 

\begin{table}[htbp] 
\centering
\caption{Parameters for the base model.}\label{base-parameters}
\begin{threeparttable} 
\begin{tabular}{@{}ll|ll@{}} 
\toprule
&\textbf{Model parameter} & \textbf{Value}\\
\midrule
&Initial wealth & $x = 1$& \\ 
&Time horizon & $T=10$ &\\ 
&Number of time steps & $n=10$ &\\ 
&Claim arrival intensity & $\lambda = 1.0$ &\\ 
&Expected claim size & $\mu = 1.0$ &\\ 
&Safety loading of insurer & $\eta = 0.50$ &\\ 
&Safety loading of reinsurer & $\theta = 0.70$& \\ 
&Risk aversion coefficient & $\alpha = 0.30$& \\ 
&Tuning parameter & $\beta = 0.40$& \\ 
& Surrogate loss parameter & $\gamma = 10$& \\ 
&OU mean-reversion level & $\kappa = 0$& \\ 
&OU mean-reversion speed & $\xi = 0.20$& \\ 
&OU volatility coefficient & $\nu = 0.05$& \\ 
\bottomrule
\end{tabular}
\end{threeparttable}
\end{table} 

Neural network training was performed on 2000 batches, each with a batch size of $2^{14}$, using the Adam optimizer \citep{kingma15adam} with an initial learning rate of $10^{-3}$. 
The learning rate was decreased by a factor of 10 after 10 epochs without improvement, with a minimum learning rate of $10^{-5}$. 
Early stopping was employed after 20 epochs without improvement. 
Distributions, expected utilities, and ruin probabilities were computed on a test set of size $2^{25}$.

First, we want to numerically verify that the expected surrogate loss indeed approximates the ruin probability. 
To achieve this, we compute the ruin probability for our model without reinsurance (i.e.\ setting $b\equiv1$) and compare the obtained value ($\approx 34.1 \%$) with the expected surrogate loss $\mathbb{E}[g_{\gamma}(F_{b}(Y))]$ for various values of $\gamma$. 
As shown in Figure \ref{ruin-probs}, the expected surrogate loss approximates the ruin probability well when $\gamma$ is sufficiently large. 
Based on these results, we fix $\gamma=10$ for the subsequent analysis. 

\begin{figure}[htbp] 
\centering 
\captionsetup{width=0.75\textwidth,format=plain} 
\includegraphics[width=0.75\linewidth]{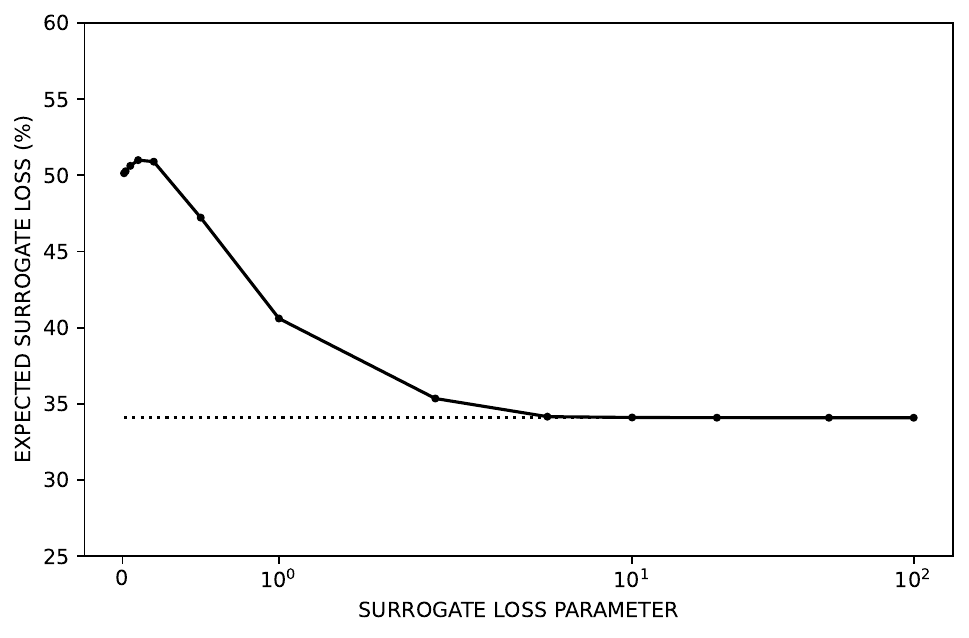} 
\caption{Expected surrogate loss $\mathbb{E}[g_{\gamma}(F_{b}(Y))]$ for various parametrizations of the surrogate loss function in the case of no reinsurance ($b \equiv 1$). The dotted horizontal line corresponds to the ruin probability ($\approx 34.1 \%$).} 
\label{ruin-probs} 
\end{figure} 

The parameter $\beta$, given in \eqref{goal}, significantly influences the optimal retention level as a function of the surplus. 
For instance, in \citet{schmidli01proportional}, where the ruin probability is minimized in a continuous-time setting, the optimal reinsurance strategy turns out to be constantly 1 until a certain surplus level is reached. 
After that, the retention level drops to a much lower value and converges to a constant as the surplus goes to infinity. 
In our model, by incorporating the ruin probability as an additional objective, the reinsurance policy must penalize scenarios where the insurer's wealth becomes negative at any point in time up to maturity $T$. 

In our simulations, we noticed that it sufficed to optimize over feedback strategies $b=(b_{i})_{i=0}^{n}$ that do not depend explicitly on time, i.e.\ $b_{i}(x) = \overline{b}(x)$ for $i=0,1,\hdots,n-1$ (with $b_{n} \equiv 1$) and for some function $\overline{b}\from \re \to [0,1]$, as allowing for time dependence did not noticeably improve the performance of the strategy. 
Therefore, we can interpret our obtained strategies as functions depending on the initial capital. 
Figure \ref{retention-levels} shows the optimal retention levels depending on the initial capital for the base case $\beta=0.4$ and the boundary cases $\beta=1$ (no penalization through ruin probability, only utility maximization) and $\beta=0$ (pure ruin probability minimization, no utility in the target functional). 

\begin{figure}[htbp] 
\centering 
\captionsetup{width=0.75\textwidth,format=plain} 
\includegraphics[width=0.75\linewidth]{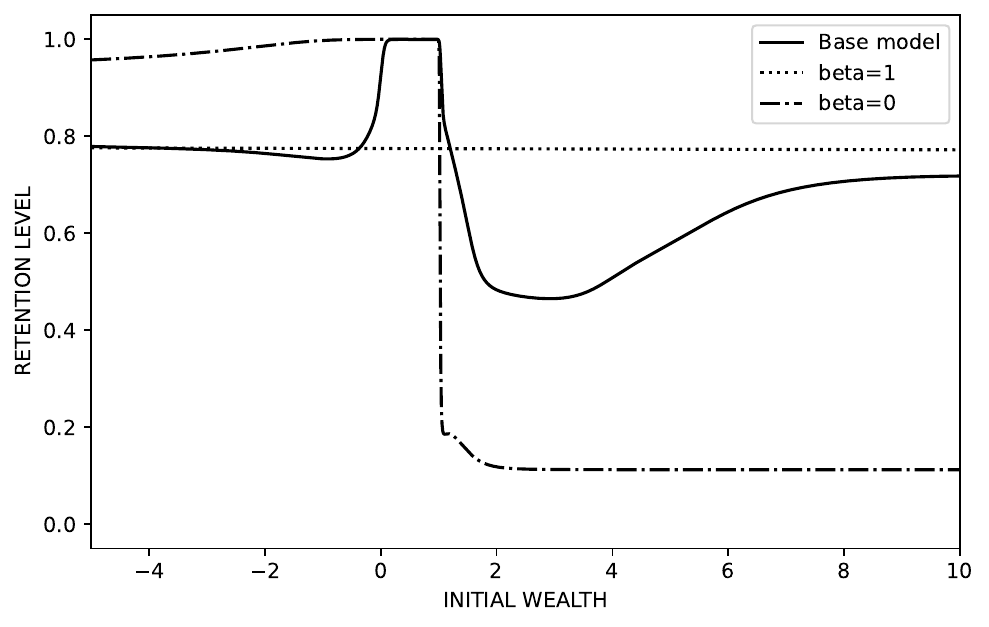}
\caption{Optimal retention levels $\overline{b}(x)$ for different values of $\beta$. The base model assumes $\beta=0.4$.} 
\label{retention-levels} 
\end{figure} 

One can check that the strategy corresponding to the case of pure ruin probability minimization has the same form (on the positive real line) as the strategy in the continuous time setting for the classical risk model, see \citet[p.~53]{schmidli08control}. 
In the base model, the presence of the utility function in the target functional enforces an increase in the optimal retention level starting from approximately $x=3$. 
In the tradeoff between more safety (a smaller retention level) and more utility, the utility wins in the long run. 
Obviously, by maximizing exponential utility with no consideration of the ruin probability, the optimal strategy does not depend on the surplus (dotted line in Figure~\ref{retention-levels}). 

Recall that our objective contains two components: the expected utility of terminal wealth, and the expected surrogate loss approximating the ruin probability. 
By varying  $\beta$, we can generate optimal reinsurance strategies that approximately achieve Pareto-efficient solutions, where any improvement in one objective requires a compromise in the other. 
It is important to note that our results are subject to numerical errors due to: 
\begin{enumerate}
\item The surrogate loss function for the ruin probability, 
\item The finite size of training and test datasets, and 
\item The fact that neural networks can in general only approximate optimal policies up to some~$\varepsilon$. 
\end{enumerate}

Figure \ref{trade-off-front} demonstrates the trade-off between expected utility of terminal wealth and ruin probability. 
Each individual point represents an approximate Pareto-efficient solution for different choices of $\beta \in [0,1]$, where improving one objective comes at the expense of the other. 
The star indicates the values obtained without reinsurance (i.e.\ $b\equiv1$), highlighting the benefits of optimal reinsurance strategies. 
In particular, the figure highlights that, for our choice of parameters, reinsurance is always more favorable than non-reinsurance. 
However, by interpreting Figure \ref{trade-off-front} one should keep in mind that the expected utility and the survival probability are not as counterbalancing as they are in a mean-variance framework. 
Indeed, negative capital is in some sense penalized by a lower utility. 

\begin{figure}[htbp] 
\centering 
\captionsetup{width=0.75\textwidth,format=plain} 
\includegraphics[width=0.75\linewidth]{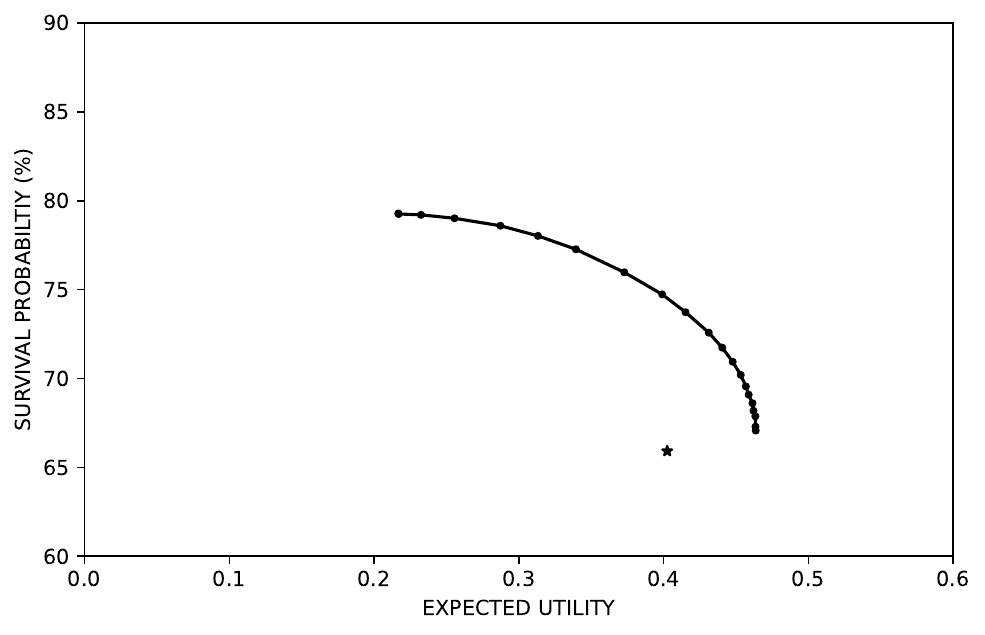}
\caption{Trade-off between expected utility of terminal wealth and survival probability obtained through optimal algorithmic reinsurance polices for different choices of $\beta \in [0,1]$. The star denotes the corresponding values obtained when no reinsurance is available.}
\label{trade-off-front} 
\end{figure} 

Finally, we investigate the effect of a time-varying exposure on the optimal retention rate. 
To this end, we replace the constant intensity $\lambda$ of the Poisson process $N$ by a deterministic, time-varying function $\lambda(t)$. 
The interpretation is, that the exposure is growing over time, i.e.\ we write $\lambda(t) = v(t) \lambda_0$, where $v(t)$ is a growing function that represents exposure, and $\lambda_0 = 1.0$ is a baseline intensity. 
Note that for this model, the total premium income increases, while the distribution of claim sizes remains unchanged. 
We assume that $v$ is normalized to 1 at time $0$, therefore starting from $v(0) = 1$, and consider 3 scenarios: (1) $v(t)$ remains constant over time, (2) $v(t)$ declines by $10\%$ per year, and (3) $v(t)$ grows by $10\%$ per year. 
Note that the first scenario corresponds to the baseline case already discussed in Figure~\ref{retention-levels} above. 

Optimizing retention levels in models with dynamic parameters represents a theoretical and numerical challenge. 
Figure~\ref{exposure-impact} highlights that the optimal strategies assume highly complex behavior. 
The case of decreasing exposure sees near full retention, except for the case when wealth remains in an interval starting just below 2 and ending just above 4. 
In contrast, the case of increasing exposure sees zero retention for negative wealth, which rapidly increases and then falls again, before gradually increasing. 
In all cases, the optimal retention levels appear to be highly non-linear, a behavior that has also been reported in \citet{hipp03dynamic}. 

\begin{figure}[htbp] 
\centering 
\captionsetup{width=0.75\textwidth,format=plain} 
\includegraphics[width=0.75\linewidth]{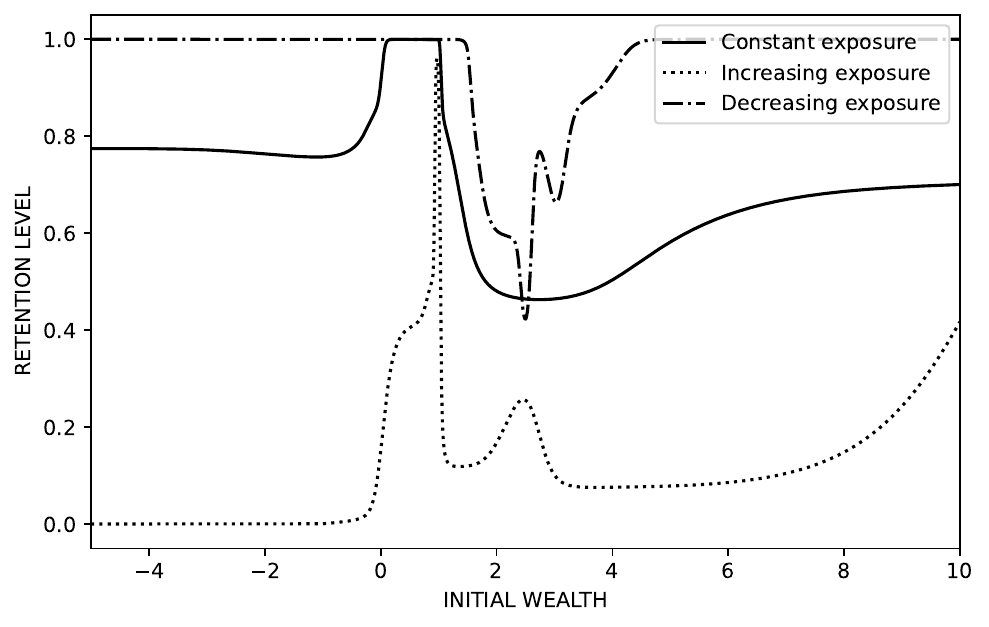}
\caption{Optimal retention levels for time-varying exposure.} 
\label{exposure-impact} 
\end{figure} 

%%%%%%%%%%%%%%%%%%%%%%%%%%%%%%----------------------------------------%%%%%%%%%%%%%%%%%%%%%%%%%%%%%% 
\section{Conclusions}\label{sec:conclusion} 
We introduce a novel framework for optimizing reinsurance strategies using a deep learning approach. 
The target functional consists of the expected utility of terminal wealth perturbed by a modified Gerber--Shiu penalty function. 
It allows to balance between maximizing the expected utility of terminal wealth, and minimizing the probability of ruin, depending on the individual preferences of the insurer. 

We draw connections to concepts from binary classification and surrogate loss functions. 
This enables the problem to be addressed using empirical risk minimization methods. 
Combined with stochastic gradient descent, it allows for efficient optimization of algorithmic reinsurance policies, even in complex model settings. 

Our numerical findings highlight the ability of our method to interpolate between the problems of maximizing the expected utility of terminal wealth, and minimizing the probability of ruin. 
Future research could explore other optimization targets and reinsurance forms as well as more complex, higher-dimensional models with correlated business lines. 
Moreover, it would be interesting to explore optimal algorithmic reinsurance strategies that are robust with respect to parameter uncertainties, such as the intensity of the Poisson process, or the expected claim size. 
Finally, another interesting topic would be to optimize algorithmic reinsurance polices under distributional constraints on the wealth process. 

\section*{Acknowledgements}
Aleksandar Arandjelovi{\'c} acknowledges support from the International Cotutelle Macquarie University Research Excellence Scholarship. 
A significant part of the work on this project was carried out while Aleksandar Arandjelovi{\'c} was affiliated with the Research Unit of Financial and Actuarial Mathematics (FAM) at TU Wien in Vienna, Austria, and the Department of Actuarial Studies and Business Analytics at Macquarie University in Sydney, Australia. 
This work was presented at the Scandinavian Actuarial Conference in August 2024 (University of Copenhagen, Denmark) and at the European Actuarial Journal Conference in September 2024 (ISEG, Lisbon, Portugal). 
The authors are grateful for the constructive comments received from participants at these events. 
The authors declare that they have no conflicts of interest. 

%%%%%%%%%%%%%%%%%%%%%%%%%%%%%%----------------------------------------%%%%%%%%%%%%%%%%%%%%%%%%%%%%%% 
\bibliography{bibliography} 

\end{document}